\documentclass{amsart}

\theoremstyle{plain}
\newtheorem{theorem}{Theorem}[section]
\newtheorem{proposition}[theorem]{Proposition}
\newtheorem{lemma}[theorem]{Lemma}
\newtheorem{corollary}[theorem]{Corollary}

\theoremstyle{definition}
\newtheorem{definition}[theorem]{Definition}
\newtheorem{example}[theorem]{Example}

\makeatletter
\@addtoreset{equation}{section} \makeatother

\begin{document}

\title{Investigation into the role of the Laurent property in integrability}
\author[T. Mase]{Takafumi Mase}
\address{Graduate School of Mathematical Sciences, 
the University of Tokyo, 3-8-1 Komaba, Meguro-ku, Tokyo 153-8914, Japan.}

\begin{abstract}
We study the Laurent property for autonomous and nonautonomous discrete equations.
First we show, without relying on the caterpillar lemma, the Laurent property for the Hirota-Miwa and the discrete BKP equations.
Next we introduce the notion of reductions and gauge transformations for discrete bilinear equations and we prove that these preserve the Laurent property.
Using these two techniques, we obtain the explicit condition on the coefficients of a nonautonomous discrete bilinear equation for it to possess the Laurent property.
Finally we study the denominators of the iterates of an equation with the Laurent property and we show that any reduction to a mapping on a one-dimensional lattice of a nonautonomous Hirota-Miwa equation or discrete BKP equation, with the Laurent property, has zero algebraic entropy.
\end{abstract}

\maketitle


\section{Introduction}\label{sec:intro}

The Laurent property is a concept that arose from the study of cluster algebras, introduced by Fomin and Zelevinsky \cite{cluster1} and recently used in many mathematical fields.
A cluster algebra is a commutative ring with some characteristic generators, which are called cluster variables.
It is known that any cluster variable can be expressed as a Laurent polynomial of the initial cluster variables.
This is called the {\em Laurent phenomenon} or the {\em Laurent property.}

In this paper we shall consider the Laurent property for discrete dynamical systems.
The natural definition of this property is as follows.

\begin{definition}[Laurent property]
An initial value problem of a discrete system has the Laurent property if every iterate can be expressed as a Laurent polynomial of the initial values.
\end{definition}

\begin{example}\label{exa:ex1}
Consider the recurrence
\[
\left\{\begin{array}{l}
	f_{m} = \dfrac{f^2_{m-1} + \beta}{f_{m-2}}, \\
	f_0 = X, f_1 = Y,
\end{array}\right.
\]
where $\beta$ is a nonzero parameter.

The Laurent property of this system can be easily shown.
It is clear that $f_2$ and $f_3$ are Laurent polynomials of $X$ and $Y$.
For $m \ge 4$, the relation $f_{m-1}f_{m-3} = f^2_{m-2} + \beta$ implies the coprimeness of $f_{m-2}$ and $f_{m-3}$ as Laurent polynomials since if $g$ divides $f_{m-2}$ and $f_{m-3}$, then $\beta \equiv 0 \mod g$.
Thus the second equality in
\[
	f_m = \frac{f^2_{m-1} + \beta}{f_{m-2}} = \frac{f^3_{m-2} + 2 \beta f_{m-2} + \beta f_{m-4}}{f^2_{m-3}}
\]
shows that $f_m$ is a Laurent polynomial.
\end{example}

\begin{example}
Consider the lattice equation, which is in fact a discrete form of the Liouville equation \cite{hirotaliou}.
\[
\left\{\begin{array}{l}
	f_{lm} = \dfrac{f_{l, m-1}f_{l-1,m} + \beta}{f_{l-1, m-1}}, \\
	f_{l0} = X_{l0}, f_{0m} = X_{0m},
\end{array}\right.
\]
where $\beta$ is a nonzero parameter.
We can show the Laurent property of this system in the same way as above.
\end{example}

In \cite{fz}, Fomin and Zelevinsky studied several recurrence relations for which they showed the Laurent property by using the so-called ``caterpillar lemma,'' which is used to prove the Laurent phenomenon for cluster algebras.
In \cite{rims}, we have discussed, without proofs, a close relation between the Laurent property and discrete integrable systems.
One aim of this paper is to give proofs of the theorems in \cite{rims}, proofs, it must be stressed, that do not make use of the caterpillar lemma.

Another aim is to discuss the Laurent property for nonautonomous systems.
As shown in the following example, nonautonomous equations can have the Laurent property when their parameters satisfy some specific relations.

\begin{example}\label{exa:ex3}
Let $r$ be a positive integer and $\alpha_m, \beta_m$ be nonzero parameters depending on $m$.
Then, the equation
\[
\left\{\begin{array}{l}
	f_{m} = \dfrac{\alpha_m f^r_{m-1} + \beta_m}{f_{m-2}}, \\
	f_0 = X, f_1 = Y
\end{array}\right.
\]
has the Laurent property if and only if $\alpha_m$ and $\beta_m$ satisfy
\[
	\alpha_m \alpha_{m-2}\beta^r_{m-1} = \beta_m \beta_{m-2}
\]
for $m \ge 4$.
This relation arises from the following calculation:
\[
	\alpha_m f^r_{m-1} + \beta_m \equiv \frac{\alpha_m \alpha_{m-2} \beta^r_{m-1} - \beta_m \beta_{m-2}}{\alpha_{m-2}f^r_{m-3}} \mod f_{m-2}.
\]
\end{example}

Among the wide variety of possible nonautonomous discrete systems, we shall mainly study so-called bilinear systems \cite{hirotakdv, hirotatoda}.
In \textsection\ref{sec:dbil}, we give a brief introduction to these systems and their initial value problems.
Some notations used throughout the paper, which are convenient when we consider the Laurent property and reductions, are also introduced.
At the end of the section, we give elementary proofs of the Laurent property of the autonomous Hirota-Miwa equation and the autonomous discrete BKP equation, which do not rely on the caterpillar lemma.

Discrete bilinear equations are usually studied in the context of integrable systems, and most of them are obtained as reductions of the Hirota-Miwa equation or the discrete BKP equation.
In \textsection\ref{sec:reduc}, we introduce the notion of a reduction of a nonautonomous bilinear equation and we show that the Laurent property is preserved by reductions (Proposition~\ref{prop:reduc}).

Gauge transformations, i.e.\ transformations in which the dependent variable is multiplied by a non-vanishing global function, can be applied to nonautonomous systems in order to change the parametric dependence of the equation.
In \textsection\ref{sec:gauge}, we introduce gauge transformations and show that they preserve the Laurent property as well (Proposition~\ref{prop:gauge}).

In \textsection\ref{sec:main}, we give the general conditions on the coefficients of the nonautonomous Hirota-Miwa equation (and the nonautonomous discrete BKP equation) for it to have the Laurent property and we show that these coincide with necessary and sufficient conditions for the equation to be transformable to an autonomous system by gauge transformation (Theorems~\ref{thm:mainhm} and~\ref{thm:mainbkp}).

The denominator of the solution to an equation with the Laurent property is easy to investigate since it is a monomial.
In \textsection\ref{sec:deg}, we study the structure of such denominators and we discuss the algebraic entropy of the corresponding equations, especially in the case of discrete bilinear equations.

Throughout the paper, except in \textsection\ref{sec:deg}, we consider all equations over a base field $k$ with an arbitrary characteristic.
While we consider only the case $k = \mathbb{R}$ in \textsection\ref{sec:deg} for technical reasons, the final result (Theorem~\ref{thm:ent}) still holds over any base field.


\section{Discrete bilinear equations}\label{sec:dbil}

In this section, we introduce what are called discrete bilinear equations \cite{hirotakdv, hirotatoda, miwa}, as well as their initial value problems.
Since a detailed description of the latter is provided in \cite{rims}, we give only the essentials.

Let $L$ be a lattice (free $\mathbb{Z}$ module of finite rank), $\alpha^{(i)}_h \in k^{\times}$ parameters, and let $v_i, u_i, w \in L$, such that $v_i + u_i = w,$ generate $L$ as a lattice.
Then, most of the discrete bilinear equations that appear in the literature can be written as
\begin{equation}\label{eq:dbil}
	f_h = \frac{\alpha^{(1)}_h f_{h+v_1}f_{h+u_1} + \cdots + \alpha^{(n)}_h f_{h+v_n}f_{h+u_n}}{f_{h+w}}.
\end{equation}
For this equation to define a proper evolution it is necessary to require $v_i, u_i$ to be $\mathbb{Z}_{\ge 0}$-linearly independent.
Moreover, we demand that $v_i \ne v_j, u_i \ne u_j, v_i \ne u_j$ for any $i \ne j$, in case we choose to decrease the number of terms of the equation.

\begin{example}[Hirota-Miwa equation \cite{hirota, miwa}]\label{exa:hm}
The autonomous Hirota-Miwa equation (the discrete KP equation) is usually written as
\[
	\alpha f_{l-1,m,n} f_{l,m-1,n-1} + \beta f_{l,m-1,n} f_{l-1,m,n-1} + \gamma f_{l,m,n-1} f_{l-1,m-1,n} = 0,
\]
where $f$ is the dependent variable, $l, m, n$ are the independent variables and $\alpha, \beta, \gamma$ are nonzero parameters.
Taking $L = \mathbb{Z}^3, h = (l, m, n), v_1 = (-1, 0, 0), u_1 = (0, 1, -1), v_2 = (0, 0, -1), u_2 = (-1, 1, 0), w = (-1, 1, -1)$ and redefining the parameters appropriately, the equation can be written as
\[
	f_h = \frac{\alpha f_{h+v_1}f_{h+u_1} + \beta f_{h+v_2}f_{h+u_2}}{f_{h+w}}.
\]
Its nonautonomous form is
\[
	f_h = \frac{\alpha_h f_{h+v_1}f_{h+u_1} + \beta_h f_{h+v_2}f_{h+u_2}}{f_{h+w}},
\]
which we shall call the nonautonomous Hirota-Miwa equation.
\end{example}

\begin{example}[discrete BKP equation]\label{exa:bkp}
The autonomous discrete BKP equation (the Miwa equation) is written as follows \cite{miwa}:
\[
	\alpha f_{lmn}f_{l-1,m-1,n-1} + \beta f_{l-1,m,n} f_{l,m-1,n-1} + \gamma f_{l,m-1,n} f_{l-1,m,n-1} + \delta f_{l,m,n-1} f_{l-1,m-1,n} = 0.
\]
In a similar way as above, we can rewrite the equation in the form
\[
	f_h = \frac{\alpha f_{h+v_1} f_{h+u_1} + \beta f_{h+v_2} f_{h+u_2} + \gamma f_{h+v_3} f_{h+u_3}}{f_{h+w}},
\]
where $v_1 = (-1, 0, 0), u_1 = (0, -1, -1), v_2 = (0, -1, 0), u_2 = (-1, 0, -1), v_3 = (0, 0, -1), u_3 = (-1, -1, 0), w = (-1, -1, -1)$.
Its nonautonomous form is
\begin{equation}\label{eq:dbkp}
	f_h = \frac{\alpha_h f_{h+v_1} f_{h+u_1} + \beta_h f_{h+v_2} f_{h+u_2} + \gamma_h f_{h+v_3} f_{h+u_3}}{f_{h+w}},
\end{equation}
which we call the nonautonomous discrete BKP equation.
\end{example}

\begin{example}[discrete KdV equation]\label{exa:dkdv}
The bilinear form of the autonomous discrete KdV equation is written as follows \cite{hirotakdv}:
\[
	(1 + \delta)f_{t,n}f_{t-2,n-1} - f_{t-1,n}f_{t-1,n-1} - \delta f_{t,n-1}f_{t-2,n} = 0.
\]
In a similar way as above, we have the nonautonomous discrete KdV equation:
\[
	f_h = \frac{\alpha_h f_{h+v_1}f_{h+u_1} + \beta_h f_{h+v_2}f_{h+u_2}}{f_{h+w}},
\]
with $L = \mathbb{Z}^2, v_1 = (-1, 0), u_1 = (-1, -1), v_2 = (0, -1), u_2 = (-2, 0), w = (-2, -1)$.
\end{example}

In order to study the Laurent property of an equation on a multi-dimensional lattice, it is necessary to introduce an appropriate initial value problem.
Following \cite{fz},

\begin{definition}[good domain]
A nonempty subset $H \subset L$ is called a good domain if it satisfies the following two conditions:
\begin{itemize}
\item
For any $h \in H$, the set
\[
	 H \cap \{ h + a_1 v_1 + \cdots + a_n v_n + b_1 u_1 + \cdots + b_n u_n \, | \, a_i, b_i \in \mathbb{Z}_{\ge 0} \}
\]
is finite.
\item
$h - \sum a_i v_i - \sum b_i u_i \in H$ for all $h \in H$ and for arbitrary nonnegative integers $a_i, b_i$.
\end{itemize}
When $H$ is a good domain, we shall call
\[
	H_0 = \{ h \in H \, | \, \text{ at least one of the } h + v_i, h + u_i, h + w \text{ do not belong to } H \}
\]
the initial domain for $H$.
It immediately follows from the conditions on a good domain that $H_0$ coincides with $\{ h \in H \, | \, h + w \notin H \}$.
\end{definition}

\begin{example}
$H = \{ (l, m, n) \in \mathbb{Z}^3 \, | \, l, m, n \ge 0 \}$ is a good domain for the Hirota-Miwa and the discrete BKP equations.
The initial domain for $H$ is $\{ (l, m, n) \in H \, | \, lmn = 0 \}$.
\end{example}

\begin{definition}[Laurent property]
Let $H \subset L$ be a good domain and consider the initial value problem
\[
	f_h = \begin{cases}
		\dfrac{\alpha^{(1)}_h f_{h+v_1}f_{h+u_1} + \cdots + \alpha^{(n)}_h f_{h+v_n}f_{h+u_n}}{f_{h+w}} & (h \in H \setminus H_0) \\
		X_h & (h \in H_0).
	\end{cases}
\]
This initial value problem has the Laurent property if all $f_h$ are Laurent polynomials of the initial values $X_{h_0} (h_0 \in H_0)$.
A discrete bilinear equation has the Laurent property if for any good domain $H \subset \mathbb{Z}$, the corresponding initial value problem has the Laurent property.
\end{definition}

It should be noted that the autonomous Hirota-Miwa equation is essentially the same as the ``octahedron recurrence'' in \cite{fz}, the autonomous discrete BKP equation as the ``cube recurrence'' and the autonomous discrete KdV equation as the ``knight recurrence'' in that reference.
Since the definitions concerning domains are essentially those of \cite{fz}, the facts proved there can be rephrased as follows.

\begin{theorem}[Fomin-Zelevinski \cite{fz}]\label{thm:fz}
In the autonomous cases, the Hirota-Miwa equation, the discrete BKP equation and the discrete KdV equation have the Laurent property.
\end{theorem}

Since the caterpillar lemma which is used in \cite{fz} to show the theorem is powerful but complicated, we prefer to give an elementary proof in the case of the Hirota-Miwa equation and the discrete BKP equation.
The Laurent property of the discrete KdV equation will follow from Proposition~\ref{prop:reduc}.

Before proving Theorem~\ref{thm:fz}, we review some basic knowledge of rings and we introduce a semi-order on $H$.

\begin{definition}
Let $A$ be a UFD (unique factorization domain).
A non-unit $g \in A$ is irreducible if $g$ cannot be written as a product of two non-units.
Two elements $f, g \in A$ are coprime if they have no common factors besides units.
Note that if $f$ and $g$ are irreducible then they are coprime, unless there exists a unit $x$ such that $f = xg$.
\end{definition}

\begin{definition}\label{defi:order}
Consider a discrete bilinear equation of the form (\ref{eq:dbil}).
Let $H \subset L$ be a good domain and define a semi-order $\le$ on $H$ by
\[
	h_1 \le h_2 \quad
	\Leftrightarrow \quad
	h_1 - h_2 \in \operatorname{span}_{\mathbb{Z}_{\ge 0}}(v_i, u_i).
\]
The first condition for a good domain guarantees that any nonempty subset of $H$ has a (possibly non-unique) minimal element.
Therefore we can use induction on $h \in H$ with respect to $\le$.
\end{definition}

\begin{example}
In the case of the discrete BKP equation (\ref{eq:dbkp}), the semi-order $\le$ can be written as follows:
\[
	h_1 = (l_1, m_1, n_1) \le h_2 = (l_2, m_2, n_2) \quad
	\Leftrightarrow \quad
	l_1 \le l_2, m_1 \le m_2, n_1 \le n_2.
\]
\end{example}

\begin{proof}[Elementary Proof of Theorem~\ref{thm:fz} (case of the Hirota-Miwa equation)]

Let $H \subset \mathbb{Z}^3$ be a good domain and consider the initial value problem
\[
	f_h = \begin{cases}
	\dfrac{\alpha f_{h+v_1}f_{h+u_1} + \beta f_{h+v_2}f_{h+u_2}}{f_{h+w}}	&	(h \in H \setminus H_0), \\
	X_h	&	(h \in H_0).
	\end{cases}
\]
Let $\le$ be the semi-order defined in Definition~\ref{defi:order} and let $A = k[X_{h_0}, X^{-1}_{h_0} \, | \, h_0 \in H_0]$ be the Laurent polynomial ring of initial values.

Let us show the following two statements by induction on $h \in H$:
\begin{itemize}
\item[(1)]
$f_h \in A$.
\item[(2)]
If $h \in H \setminus H_0$, then $f_h$ is an irreducible element of $A$.
\end{itemize}

First we show (1).
Since the case $h + w \in H_0$ is trivial, we may assume that $h + w \in H \setminus H_0$.
Let $F = \alpha f_{h+v_1} f_{h+u_1} + \beta f_{h+v_2} f_{h+u_2}$ and
\[
	A' = A[f^{-1}_{h+w+v_1}, f^{-1}_{h+w+u_1}, f^{-1}_{h+w+v_2}, f^{-1}_{h+w+u_2}] / (f_{h+w}).
\]
By the induction hypothesis (2), $f_{h+w}, f_{h+w+v_1}, f_{h+w+u_1}, f_{h+w+v_2}, f_{h+w+u_2}$ are irreducible in $A$ and pairwise coprime.
Thus, it is sufficient to show that $F = 0$ in the quotient ring $A'$.
Using
\[
	v_1 + u_1 = v_2 + u_2 = w,
\]
we have
\begin{align*}
	f_{h+v_1} &= \frac{\beta f_{h+v_1+v_2}f_{h+v_1+u_2}}{f_{h+w+v_1}}, &
	f_{h+u_1} &= \frac{\beta f_{h+u_1+v_2}f_{h+u_1+u_2}}{f_{h+w+u_1}}, \\
	f_{h+v_2} &= \frac{\alpha f_{h+v_1+v_2}f_{h+u_1+v_2}}{f_{h+w+v_2}}, &
	f_{h+u_2} &= \frac{\alpha f_{h+v_1+u_2}f_{h+u_1+u_2}}{f_{h+w+u_2}},
\end{align*}
and
\begin{align*}
	F	&=	\alpha\beta^2\frac{f_{h+v_1+v_2}f_{h+v_1+u_2}f_{h+u_1+v_2}f_{h+u_1+u_2}}{f_{h+w+v_1}f_{h+w+u_1}} + \alpha^2\beta\frac{f_{h+v_1+v_2}f_{h+v_1+u_2}f_{h+u_1+v_2}f_{h+u_1+u_2}}{f_{h+w+v_2}f_{h+w+u_2}} \\
	&=	\frac{\alpha\beta f_{h+v_1+v_2}f_{h+v_1+u_2}f_{h+u_1+v_2}f_{h+u_1+u_2}}{f_{h+w+v_1}f_{h+w+u_1}f_{h+w+v_2}f_{h+w+u_2}}\left(\alpha f_{h+w+v_1}f_{h+w+u_1} + \beta f_{h+w+v_2}f_{h+w+u_2}\right) \\
	&=	\frac{\alpha\beta f_{h+v_1+v_2}f_{h+v_1+u_2}f_{h+u_1+v_2}f_{h+u_1+u_2}}{f_{h+w+v_1}f_{h+w+u_1}f_{h+w+v_2}f_{h+w+u_2}}f_{h+w}f_{h+2w} \\
	&=	0.
\end{align*}

Next, we show (2).
Let $m$ be the smallest integer satisfying $h + mv_1 \in H_0$.
An easy induction on $m$ shows that $f_h$ can be represented as
\[
	f_h = \frac{\alpha^m f_{h+u_1}}{X_{h+mv_1+u_1}}X_{h+mv_1} + g,
\]
where $g$ and $f_{h+u_1}$ do not depend on $X_{h+mv_1}$.
If $h + u_1 \in H_0$, then $\alpha^m f_{h+u_1} / X_{h+mv_1+u_1}$ is a unit in $A$ and thus $f_h$ is irreducible.
On the other hand, if $h + u_1 \in H \setminus H_0$, the induction hypothesis (2) implies the irreducibility of $f_{h+u_1}$.
Thus, it is sufficient to show that $f_h$ cannot be divided by $f_{h+u_1}$.
Let $m'$ be the smallest integer satisfying $h + m'v_2 \in H_0$.
Then, $f_h$ can be represented as
\[
	f_h = \frac{\beta^{m'} f_{h+u_2}}{X_{h+m'v_2+u_2}}X_{h+m'v_2} + g',
\]
where $g'$ and $f_{h+u_2}$ do not depend on $X_{h+m'v_2}$.
Since $f_{h+u_1}$ and $f_{h+u_2}$ are coprime, $f_h$ is not divisible by $f_{h+u_1}$.
\end{proof}

\begin{proof}[Elementary Proof of Theorem~\ref{thm:fz} (the case of discrete BKP equation)]

Let $H \subset \mathbb{Z}^3$ be a good domain and consider the corresponding initial value problem
\[
	f_h = \begin{cases}
	\dfrac{\alpha f_{h+v_1}f_{h+u_1} + \beta f_{h+v_2}f_{h+u_2}{f_{h+w}} + \gamma f_{h+v_3}f_{h+u_3}}{f_{h+w}}	&	(h \in H \setminus H_0), \\
	X_h	&	(h \in H_0).
	\end{cases}
\]
The only difference with the case of the Hirota-Miwa equation is the detail of the proof of (1).

Let
\[
	F = \alpha f_{h+v_1}f_{h+u_1} + \beta f_{h+v_2}f_{h+u_2}{f_{h+w}} + \gamma f_{h+v_3}f_{h+u_3}
\]
and
\[
	A' = A[f^{-1}_{h+w+v_1}, f^{-1}_{h+w+v_2}, f^{-1}_{h+w+v_3}, f^{-1}_{h+w+u_1}, f^{-1}_{h+w+u_2}, f^{-1}_{h+w+u_3}] / (f_{h+w}).
\]
As in the case of the Hirota-Miwa equation, it is sufficient to show that $F = 0$ in the ring $A'$.

Using
\[
	v_1 + u_1 = v_2 + u_2 = v_3 + u_3 = v_1 + v_2 + v_3 = w,
\]
we have
\[
	f_{h+u_1} = \frac{\beta f_{h+u_1+v_2}f_{h+u_1+u_2} + \gamma f_{h+u_1+v_3}f_{h+u_1+u_3}}{f_{h+w+u_1}}
\]
and
\begin{align*}
	f_{h+v_1}	&=	\frac{\beta f_{h+v_1+v_2}f_{h+v_1+u_2} + \gamma f_{h+v_1+v_3}f_{h+v_1+v_3}}{f_{h+w+v_1}} \\
	&=	\frac{\beta f_{h+v_1+u_2}}{f_{h+w+v_1}}f_{h+u_3} + \frac{\gamma f_{h+v_1+u_3}}{{f_{h+w+v_1}}}f_{h+u_2} \\
	&=	\frac{\beta f_{h+v_1+u_2}}{f_{h+w+v_1}f_{h+w+u_3}}(\alpha f_{h+v_1+u_3}f_{h+u_1+u_3} + \beta f_{h+v_2+u_3}f_{h+u_2+u_3}) \\
	& \quad + \frac{\gamma f_{h+v_1+u_3}}{f_{h+w+v_1}f_{h+w+u_2}}(\alpha f_{h+v_1+u_2}f_{h+u_1+u_2} + \gamma f_{h+v_3+u_2}f_{h+u_2+u_3}) \\
	&=	\frac{\beta^2 f_{h+v_1+u_2}f_{h+v_2+u_3}}{f_{h+w+u_3}} + \frac{\gamma^2 f_{h+v_1+u_3}f_{h+v_3+u_2}}{f_{h+w+u_2}} \\
	& \quad + \frac{\alpha f_{h+v_1+u_2}f_{h+v_1+u_3}}{f_{h+w+v_1}f_{h+w+u_2}f_{h+w+u_3}}(\beta f_{h+w+v_2}f_{h+w+u_2} + \gamma f_{h+w+v_3}f_{h+w+u_3}) \\
	&=	\frac{\beta^2 f_{h+v_1+u_2}f_{h+v_2+u_3}f_{h+w+u_2} + \gamma^2 f_{h+v_1+u_3}f_{h+v_3+u_2}f_{h+w+u_3}
	- \alpha^2 f_{h+v_1+u_2}f_{h+v_1+u_3}f_{h+w+u_1}}{f_{h+w+u_2}f_{h+w+u_3}}.
\end{align*}
Therefore the first term of $F$ is
\begin{multline*}
	\alpha f_{h+v_1}f_{h+u_1} = \frac{\alpha}{f_{h+w+u_1}f_{h+w+u_2}f_{h+w+u_3}}\times \\
	\bigg(\beta^3 f_{h+v_1+u_2}f_{h+v_2+u_1}f_{h+v_2+u_3}f_{h+w+v_3}f_{h+w+u_2}
	+ \gamma^3 f_{h+v_1+u_3}f_{h+v_3+u_1}f_{h+v_3+u_2}f_{h+w+v_2}f_{h+w+u_3} \\
	+ \beta^2\gamma f_{h+v_1+u_2}f_{h+v_2+u_3}f_{h+v_3+u_1}f_{h+w+v_2}f_{h+w+u_2}
	+ \beta\gamma^2 f_{h+v_1+u_3}f_{h+v_2+u_1}f_{h+v_3+u_2}f_{h+w+v_3}f_{h+w+u_3} \\
	- \alpha^2\beta f_{h+v_1+u_2}f_{h+v_2+u_1}f_{h+v_1+u_3}f_{h+w+v_3}f_{h+w+u_1}
	- \alpha^2\gamma f_{h+v_1+u_3}f_{h+v_3+u_1}f_{h+v_1+u_2}f_{h+w+v_2}f_{h+w+u_1}\bigg).
\end{multline*}
The other terms are obtained by cyclic permutation of the indices: $1 \to 2 \to 3 \to 1$, $\alpha \to \beta \to \gamma \to \alpha$.
Summing all three terms we obtain
\begin{multline*}
	F = \frac{\alpha\beta\gamma (f_{h+v_1+u_2}f_{h+v_2+u_3}f_{h+v_3+u_1} + f_{h+v_1+u_3}f_{h+v_2+u_1}f_{h+v_3+u_2})}{f_{h+w+u_1}f_{h+w+u_2}f_{h+w+u_3}} \times \\
\Big(\beta f_{h+w+v_2}f_{h+w+u_2} + \gamma f_{h+w+v_3}f_{h+w+u_3} + \alpha f_{h+w+v_1}f_{h+w+u_1}\Big).
\end{multline*}
Since
\begin{align*}
	\beta f_{h+w+v_2}f_{h+w+u_2} + \gamma f_{h+w+v_3}f_{h+w+u_3} + \alpha f_{h+w+v_1}f_{h+w+u_1}	&=	f_{h+w}f_{h+2w} \\
	&=	0,
\end{align*}
finally we have $F = 0$ in $A'$.
\end{proof}

The following corollary follows from the above proofs.

\begin{corollary}
In the case of the autonomous Hirota-Miwa equation and the autonomous discrete BKP equation, every iterate (except the initial values themselves which are units) is irreducible as a Laurent polynomial of the initial values.
\end{corollary}


\section{Reductions}\label{sec:reduc}

A reduction of a discrete bilinear equation is an equation defined on a lower dimensional lattice, obtained by requiring that its solutions be invariant under translations in some direction on the original lattice.

\begin{example}[Hirota-Miwa equation $\to$ discrete KdV equation]\label{ex:hmkdv}
Consider the Hirota-Miwa equation
\[
	f_h = \frac{\alpha_h f_{h+v_1}f_{h+u_1} + \beta_h f_{h+v_2}f_{h+u_2}}{f_{h+w}}.
\]

First, we consider the autonomous case.
Let $\alpha_h = \alpha, \beta_h = \beta$ be constants.
Let $x = v_1 + v_2 - u_1$ and assume that $f_h$ satisfy $f_h = f_{h+x}$ for all $h \in H$.
Then we have $f_{h+u_1} = f_{h+v_1+v_2}, f_{h+u_2} = f_{h+2v_1}$ and $f_{h+w} = f_{h+2v_1+v_2}$.
Let $L'$ be the sublattice of $\mathbb{Z}^3$ spanned by $v_1, v_2$, and $v'_1 = v_1, u'_1 = v_1 + v_2, v'_2 = v_2, u'_2 = 2v_1, w' = 2v_1 + v_2$.
Then, $f_h$ satisfies the discrete KdV equation
\[
	f_h = \frac{\alpha f_{h+v'_1}f_{h+u'_1} + \beta f_{h+v'_2}f_{h+u'_2}}{f_{h'+w'}}
\]
on $L'$.
All the values of $f_h$ on $L$ are recovered from the relation $f_h = f_{h+x}$.

In this procedure, the choice of $L'$ is not essential at all.
Therefore, it is better to consider the discrete KdV equation on the lattice $L / \mathbb{Z}x$ rather than on $L'$.
The discrete KdV equation on $L / \mathbb{Z}x$ is
\[
	f_{\bar{h}} = \frac{\alpha f_{\bar{h}+\bar{v}_1}f_{\bar{h}+\bar{u}_1} + \beta f_{\bar{h}+\bar{v}_2}f_{\bar{h}+\bar{u}_2}}{f_{\bar{h}+\bar{w}}},
\]
where $\bar{y} \in L / \mathbb{Z}x$ is an equivalence class of $y \in L$.
These two discrete KdV equations are essentially the same.
Clearly, taking the isomorphism $L' \to L / \mathbb{Z}x; v_1 \mapsto \bar{v}_1, v_2 \mapsto \bar{v}_2$, we can transfer the first equation to the other.

Next, we consider the nonautonomous case.
In the same way as in the autonomous case, the nonautonomous discrete KdV equation on $L / \mathbb{Z}x$ should be 
\[
	f_{\bar{h}} = \frac{\bar{\alpha}_{\bar{h}} f_{\bar{h}+\bar{v}_1}f_{\bar{h}+\bar{u}_1} + \bar{\beta}_{\bar{h}} f_{\bar{h}+\bar{v}_2}f_{\bar{h}+\bar{u}_2}}{f_{\bar{h}+\bar{w}}},
\]
where $\bar{\alpha}_{\bar{h}} = \alpha_h$ and $\bar{\beta}_{\bar{h}} = \beta_h$.
However, $\bar{\alpha}_{\bar{h}}$ and $\bar{\beta}_{\bar{h}}$ are not well-defined in  general.
To guarantee well-definedness the conditions $\alpha_h = \alpha_{h+x}$ and $\beta_h = \beta_{h+x}$ are necessary.
If these conditions do not hold, the nonautonomous Hirota-Miwa equation cannot be reduced to the nonautonomous discrete KdV equation.
\end{example}

Generalizing this procedure, let us define reductions of discrete bilinear equations as follows.

\begin{definition}[reduction]
A surjective $\mathbb{Z}$-linear map $\varphi \colon L \to L'$ is a reduction of a discrete bilinear equation (\ref{eq:dbil}) if the following three conditions hold:
\begin{itemize}
\item
$\alpha^{(i)}_h$ is ($\ker \varphi$)-invariant, i.e.\ $\alpha^{(i)}_h = \alpha^{(i)}_{h + x}$ for all $h \in H$ and $x \in \ker \varphi$.

\item
$\varphi(v_i), \varphi(u_i)$ are $\mathbb{Z}_{\ge 0}$-linearly independent.

\item
$\varphi(v_i) \ne \varphi(v_j), \varphi(u_i) \ne \varphi(u_j), \varphi(v_i) \ne \varphi(u_j)$ for all $i \ne j$.

\end{itemize}
The equation on $L'$ obtained by a reduction $\varphi$ is
\[
	f'_{h'} = \frac{\alpha'^{(1)}_{h'}f'_{h'+v'_1}f'_{h'+u'_1} + \cdots + \alpha'^{(n)}_{h'}f'_{h'+v'_n}f'_{h'+u'_n}}{f'_{h'+w'}},
\]
where $y' = \varphi(y)$ for $y \in L$ and $\alpha'^{(i)}_{h'} = \alpha^{(i)}_{\varphi^{-1}(h')}$.
The equation obtained by $\varphi$ is also called a reduction.
\end{definition}

The first condition is necessary for the well-definedness of $\alpha^{(i)}_{h'}$.
The second condition is equivalent to the $\mathbb{R}_{\ge 0}$-linear independence of the lattice points $\varphi(v_i), \varphi(u_i)$.
The third condition is such that a reduction does not decrease the number of terms in the equation.

The following proposition gives the conditions for equations to be reductions of the Hirota-Miwa equation or the discrete BKP equation.

\begin{proposition}
\begin{itemize}
\item[(1)]
Every discrete bilinear equation of the form (\ref{eq:dbil}) with three terms
\[
	f'_{h'} = \frac{\alpha'_{h'} f'_{h'+v'_1} f'_{h'+u'_1} + \beta'_{h'} f'_{h'+v'_2} f'_{h'+u'_2}}{f'_{h'+w'}}
\]
can be obtained as a reduction of the Hirota-Miwa equation.

\item[(2)]
A discrete bilinear equation of the form (\ref{eq:dbil}) with four terms
\begin{equation}\label{eq:bkpreduc}
	f'_{h'} = \frac{\alpha'_{h'} f'_{h'+v'_1} f'_{h'+u'_1} + \beta'_{h'} f'_{h'+v'_2} f'_{h'+u'_2} + \gamma'_{h'} f'_{h'+v'_3} f'_{h'+u'_3}}{f'_{h'+w'}}
\end{equation}
can be obtained as a reduction of the discrete BKP equation if and only if there exist three mutually distinct shifts $x, y, z \in \{ v'_1, v'_2, v'_3, u'_1, u'_2, u'_3 \}$ such that $x + y + z = w'$.
\end{itemize}
\end{proposition}
\begin{proof}
(1)
Let $L'$ be the lattice on which the reduced equation is defined.
Let $v_1, v_2, u_1, u_2, w \in \mathbb{Z}^3$ be as in Example~\ref{exa:hm}.
Define a $\mathbb{Z}$-linear map $\varphi \colon \mathbb{Z}^3 \to L'$ by $v_1 \mapsto v'_1, u_1 \mapsto u'_1, v_2 \mapsto v'_2$.
Since $v'_1, u'_1, v'_2, u'_2, w'$ satisfy $v'_1 + u'_1 = v'_2 + u'_2 = w'$, we have
\begin{align*}
	\varphi(u_2) &= u'_2, \\
	\varphi(w) &= w'.
\end{align*}
Let us define $\alpha_h (h \in \mathbb{Z}^3)$ by $\alpha_h = \alpha'_{\varphi(h)}$.
Then, $\varphi$ is a reduction from the Hirota-Miwa equation.

(2)
Let $v_1, \ldots, u_3, w \in \mathbb{Z}^3$ be as in Example~\ref{exa:bkp}.

First we suppose that (\ref{eq:bkpreduc}) is a reduction of the discrete BKP equation.
Let $\varphi$ be a $\mathbb{Z}$-linear map that gives the reduction.
Then, $x = \varphi(v_1), y = \varphi(v_2), z = \varphi(v_3)$ are mutually distinct and satisfy $x + y + z = w'$.

Conversely, if $x, y, z \in \{ v'_1, v'_2, v'_3, u'_1, u'_2, u'_3 \}$ satisfy $x + y + z = w'$, then the $\mathbb{Z}$-linear map defined by $v_1 \mapsto x, v_2 \mapsto y, v_3 \mapsto z$ gives a reduction from the discrete BKP equation.
\end{proof}

The most important property of reductions in this paper is the following proposition.

\begin{proposition}\label{prop:reduc}
Let $\varphi \colon L \to L'$ be a reduction, Let $H' \subset L'$ be a good domain and consider the initial value problem
\[
	f'_{h'} = \begin{cases}
	\dfrac{\alpha'^{(1)}_{h'} f'_{h'+v'_1}f'_{h'+u'_1} + \cdots + \alpha'^{(n)}_{h'} f'_{h'+v'_n}f'_{h'+u'_n}}{f'_{h'+w'}}	&	(h' \in H' \setminus H'_0), \\
	X'_{h'}	&	(h' \in H'_0).
	\end{cases}
\]

\begin{itemize}
\item[(1)]
Let $H = \varphi^{-1}(H')$.
Then, $H \subset L$ is a good domain and $H_0 = \varphi^{-1}(H'_0)$.

\item[(2)]
Let $\alpha^{(i)}_h = \alpha'^{(i)}_{\varphi(h)}$ and consider the initial value problem
\[
	f_h = \begin{cases}
	\dfrac{\alpha^{(1)}_h f_{h+v_1}f_{h+u_1} + \cdots + \alpha^{(n)}_h f_{h+v_n}f_{h+u_n}}{f_{h+w}}	&	(h \in H \setminus H_0), \\
	X_h	&	(h \in H_0).
	\end{cases}
\]
Then, $f_h \Big|_{X_{h_0} = X'_{\varphi(h_0)}}$ coincides with $f'_{\varphi(h)}$ for all $h \in H$.
In particular, the Laurent property of discrete bilinear equations cannot be lost by reductions.
\end{itemize}
\end{proposition}

\begin{proof}
We may assume without loss of generality that $L' = L / \mathbb{Z}x$ for some $x \in L \setminus \{ 0 \}$ since every nontrivial reduction can be written as a composition of reductions of corank one.

(1)
By construction, $H$ satisfies the second condition on a good domain.
Let $h \in H$ and
\begin{align*}
	S &= \operatorname{span}_{\mathbb{Z}_{\ge 0}}(v_i, u_i) \subset L, \\
	S' &= \operatorname{span}_{\mathbb{Z}_{\ge 0}}(v'_i, u'_i) \subset L'.
\end{align*}
Then, it is sufficient to show that $H \cap (h + S)$ is a finite set.
By a translation on $L$, we have
\[
	\# H \cap (h - S) = \#(H - h) \cap S,
\]
where $\#$ represents the cardinality of the set.
Since $H - h = \varphi^{-1}(H' - \varphi(h))$ and $\varphi(S) = S'$, we have
\[
	(H - h) \cap S = \varphi^{-1}((H' - \varphi(h)) \cap S') \cap S.
\]
Since $H' \subset L'$ is a good domain, $(H' - \varphi(h)) \cap S'$ is a finite set.
Thus, it is sufficient to show that for every $z' \in S'$, the set $\varphi^{-1}(z') \cap S$ is finite.

Let $z' \in S'$ and $z \in L$ satisfy $\varphi(z) = z'$.
Assume $\varphi^{-1}(z') \cap S$ to be infinite.
Since $\varphi^{-1}(z') = z + \mathbb{Z}x$, there exist distinct integers $m_0, m_1, \cdots$ satisfying $z + m_j x \in S$.
As the sign of $x$ can be changed at will, we may assume without loss of generality that $\{ m_j \}^{\infty}_{j = 0}$ contains infinitely many positive integers.
Let $L_{\mathbb{R}} = L \otimes \mathbb{R}$ and $S_{\mathbb{R}} = \operatorname{span}_{\mathbb{R}_{\ge 0}}(v_i, u_i)$.
Then, $S_{\mathbb{R}}$ is a closed convex cone in $L_{\mathbb{R}}$.
Since $0, z, z + m_j \in S_{\mathbb{R}}$ and $S_{\mathbb{R}}$ is convex, we have
\[
	\{ sz + tx \in L_{\mathbb{R}} \, | \, 0 < s < 1, t \ge 0 \} \subset S_{\mathbb{R}}.
\]
Since $S_{\mathbb{R}}$ is closed, we have
\[
	\{ sz + tx \in L_{\mathbb{R}} \, | \, 0 \le s \le 1, t \ge 0 \} \subset S_{\mathbb{R}}
\]
and thus $x$ can be expressed as a nontrivial $\mathbb{R}_{\ge 0}$-coefficient linear combination of $v_i, u_i$.
Sending the expression by $\varphi$, we can write $\varphi(x) = 0 \in L'$ as a nontrivial $\mathbb{R}_{\ge 0}$-coefficient linear combination of $v'_i, u'_i$.
However, this contradicts the second condition of reductions.
Hence, $H \subset L$ must be a good domain.

(2) is easily shown by induction.
\end{proof}


\section{Gauge transformations}\label{sec:gauge}

As mentioned in \textsection\ref{sec:intro}, a gauge transformation is an operation in which the dependent variable is multiplied by a non-vanishing global function.

Let $f_h$ be a solution of (\ref{eq:dbil}), $\phi_h \in k^{\times}$ a function on $L$ and $\tilde{f}_h = f_h \phi_h$.
Then, $\tilde{f}_h$ satisfies the following nonautonomous equation with different parameters:
\begin{align*}
	\tilde{f}_h &= \frac{\tilde{\alpha}^{(1)}_h \tilde{f}_{h+v_1}\tilde{f}_{h+u_1} + \cdots + \tilde{\alpha}^{(n)}_h \tilde{f}_{h+v_n}\tilde{f}_{h+u_n}}{\tilde{f}_{h+w}}, \\
	\tilde{\alpha}^{(i)}_h &= \frac{\phi_h \phi_{h+w}}{\phi_{h+v_i} \phi_{h+u_i}}\alpha^{(i)}_h.
\end{align*}

It is usual to think of a gauge transformation as an operation between solutions.
However, it is also possible to interpret it as an operation that transforms parameters in a given equation.
In this paper, we adopt the latter view.

\begin{definition}
Let $\phi_h \in k^{\times}$ be a function on $L$.
A gauge transformation of (\ref{eq:dbil}) by $\phi_h$ is an operation that transforms the parameters $\alpha^{(i)}_h$ as follows:
\[
	\alpha^{(i)}_h \mapsto \frac{\phi_h \phi_{h+w}}{\phi_{h+v_i} \phi_{h+u_i}}\alpha^{(i)}_h.
\]
\end{definition}

It is clear that any gauge transformation is invertible, since the transformation by $1 / \phi_h$ gives the inverse.

The following proposition is the most important property concerning gauge transformations in this paper.

\begin{proposition}\label{prop:gauge}
For discrete bilinear equations, gauge transformations preserve the Laurent property.
\end{proposition}

\begin{proof}
Let $\phi_h \in k^{\times}$ be a gauge function and $H \subset L$ a good domain.
Suppose that the corresponding initial value problem
\[
	f_h = \begin{cases}
		\dfrac{\alpha^{(1)}_h f_{h+u_1}f_{h+u_2} + \cdots + \alpha^{(n)}_h f_{h+u_n}f_{h+u_n}}{f_{h+w}} & (h \in H \setminus H_0) \\
		X_h & (h \in H_0)
	\end{cases}
\]
has the Laurent property.
The equation obtained after a gauge transformation by $\phi_h$ is
\[
	\tilde{f}_h = \begin{cases}
		\dfrac{\tilde{\alpha}^{(1)}_h \tilde{f}_{h+u_1}\tilde{f}_{h+u_2} + \cdots + \tilde{\alpha}^{(n)}_h \tilde{f}_{h+u_n}\tilde{f}_{h+u_n}}{\tilde{f}_{h+w}} & (h \in H \setminus H_0) \\
		\tilde{X}_h & (h \in H_0),
	\end{cases}
\]
where $\tilde{\alpha}^{(i)}_h = \alpha^{(i)}_h \phi_h \phi_{h+w} / \phi_{h+v_i} \phi_{h+u_i}$.
It is clear that
\[
	f_h \Big|_{X_{h_0} = \tilde{X}_{h_0} / \phi_{h_0}} = \frac{\tilde{f}_h}{\phi_h}
\]
and thus that $\tilde{f}_h$ is a Laurent polynomial of $\tilde{X}_{h_0}$.
\end{proof}

The above correspondence is an isomorphism between two Laurent polynomial rings.
Hence, many properties are invariant under gauge transformations.
For example, the degree and the irreducibility of iterates are also preserved.

\begin{example}
The proof of Proposition~\ref{prop:gauge} is in fact valid for any type of equation.
For example, consider the equation in Example~\ref{exa:ex3}.
Let $\phi_m \in k^{\times}$ be a function on $\mathbb{Z}$.
Then, $\tilde{f}_m = \phi_m f_m$ satisfies
\begin{align*}
	\tilde{f}_{m} &= \dfrac{\tilde{\alpha}_m \tilde{f}^r_{m-1} + \tilde{\beta}_m}{\tilde{f}_{m-2}}, \\
	\tilde{\alpha}_m &= \frac{\phi_m \phi_{m-2}}{\phi^r_{m-1}} \alpha_m, \\
	\tilde{\beta}_m &= \phi_m \phi_{m-2} \beta_m.
\end{align*}
In fact, a direct calculation shows that the two conditions for the Laurent property, $\alpha_m \alpha_{m-2}\beta^r_{m-1} = \beta_m \beta_{m-2}$ and $\tilde{\alpha}_m \tilde{\alpha}_{m-2}\tilde{\beta}^r_{m-1} = \tilde{\beta}_m \tilde{\beta}_{m-2}$, are equivalent.
\end{example}


\section{The Laurent property for nonautonomous bilinear equations}\label{sec:main}

In this section, we study the conditions nonautonomous equations have to satisfy to possess the Laurent property.

\begin{theorem}\label{thm:mainhm}
The following three conditions for the nonautonomous Hirota-Miwa equation are equivalent:
\begin{itemize}
\item[(1)]
The equation has the Laurent property.
\item[(2)]
$\alpha_h$ and $\beta_h$ satisfy
\begin{equation}\label{eq:aabbhm}
	\alpha_{h}\alpha_{h+w}\beta_{h+v_1}\beta_{h+u_1} = \beta_{h}\beta_{h+w}\alpha_{h+v_2}\alpha_{h+u_2}
\end{equation}
for all $h \in \mathbb{Z}^3$.
\item[(3)]
The equation can be transformed into an autonomous system by a gauge transformation.
Moreover, any nonzero value is permitted as a parameter value in the autonomous system.
\end{itemize}
\end{theorem}

\begin{theorem}\label{thm:mainbkp}
The following three conditions for the nonautonomous discrete BKP equation are equivalent:
\begin{itemize}
\item[(1)]
The equation has the Laurent property.
\item[(2)]
$\alpha_h, \beta_h, \gamma_h$ satisfy the following three relations for all $h \in \mathbb{Z}^3$:
\begin{align*}
	\alpha_{h+v_2}\beta_{h}\gamma_{h+u_1} &= \alpha_{h+v_3}\beta_{h+u_1}\gamma_{h}, \\
	\alpha_{h+u_2}\beta_{h+v_3}\gamma_{h} &= \alpha_{h}\beta_{h+v_1}\gamma_{h+u_2}, \\
	\alpha_{h}\beta_{h+u_3}\gamma_{h+v_1} &= \alpha_{h+u_3}\beta_{h}\gamma_{h+v_2}.
\end{align*}
\item[(3)]
The equation can be transformed into an autonomous system by a gauge transformation.
Moreover, any nonzero value is permitted as a parameter value in the autonomous system.
\end{itemize}
\end{theorem}

\begin{proof}[Proof of Theorem~\ref{thm:mainhm}]

(1) $\Rightarrow$ (2).
Let $h \in \mathbb{Z}^3$ and
\[
	H = \{ h + 2w - a_1 v_1 - a_2 v_2 - b_1 u_1 - b_2 u_2 \, | \, a_i, b_i \in \mathbb{Z}_{\ge 0} \}.
\]
Since $H \subset \mathbb{Z}^3$ is a good domain, the corresponding initial value problem
\[
	f_{h'} = \begin{cases}
		\dfrac{\alpha_{h'}f_{h'+v_1}f_{h'+u_1} + \beta_{h'}f_{h'+v_2}f_{h'+u_2}}{f_{h'+w}} & (h' \in H \setminus H_0) \\
		X_{h'} & (h' \in H_0)
	\end{cases}
\]
has the Laurent property.

Let
\[
	F = \alpha_{h}f_{h+v_1}f_{h+u_1} + \beta_{h}f_{h+v_2}f_{h+u_2}
\]
and calculate $F$ modulo $f_{h+w}$.
Since $f_{h+w} \equiv 0$, we have
\begin{align*}
	f_{h+v_1} & \equiv \frac{\beta_{h+v_1}f_{h+v_1+v_2}f_{h+v_1+u_2}}{X_{h+w+v_1}}, &
	f_{h+u_1} & \equiv \frac{\beta_{h+u_1}f_{h+u_1+v_2}f_{h+u_1+u_2}}{X_{h+w+u_1}}, \\
	f_{h+v_2} & \equiv \frac{\alpha_{h+v_2}f_{h+v_1+v_2}f_{h+u_1+v_2}}{X_{h+2+v_2}}, &
	f_{h+u_2} & \equiv \frac{\alpha_{h+u_2}f_{h+v_1+u_2}f_{h+u_1+u_2}}{X_{h+w+u_2}},
\end{align*}
and
\begin{multline*}
	F \equiv \frac{f_{h+v_1+v_2}f_{h+v_1+u_2}f_{h+u_1+v_2}f_{h+u_1+u_2}}{X_{h+w+v_1}X_{h+w+u_1}X_{h+w+v_2}X_{h+w+u_2}} \times \\
	\left( \alpha_h \beta_{h+v_1} \beta_{h+u_1}X_{h+w+v_2}X_{h+w+u_2} + \beta_h \alpha_{h+v_2} \alpha_{h+u_2} X_{h+w+v_1}X_{h+w+u_1} \right).
\end{multline*}
By the Hirota-Miwa equation,
\[
	X_{h+w+v_1}X_{h+w+u_1} \equiv - \frac{\beta_{h+w}}{\alpha_{h+w}} X_{h+w+v_2}X_{h+w+u_2},
\]
and thus
\[
	F \equiv \frac{f_{h+v_1+v_2}f_{h+v_1+u_2}f_{h+u_1+v_2}f_{h+u_1+u_2}}{X_{h+w+v_1}X_{h+w+u_1}X_{h+w+v_2}X_{h+w+u_2}} \left( \alpha_h \beta_{h+v_1} \beta_{h+u_1} - \frac{\beta_h\alpha_{h+v_2}\alpha_{h+u_2}\beta_{h+w}}{\alpha_{h+w}} \right) X_{h+w+v_2}X_{h+2+u_2}.
\]
Since the initial value problem has the Laurent property, $F$ must be $0$ modulo $f_{h+w}$.
Hence we have
\[
	\alpha_h \beta_{h+v_1} \beta_{h+u_1} - \frac{\beta_h\alpha_{h+v_2}\alpha_{h+u_2}\beta_{h+w}}{\alpha_{h+w}} = 0.
\]

(2) $\Rightarrow$ (3).
Suppose that $\alpha_h$ and $\beta_h$ satisfy the relation in (2) and let $\tilde{\alpha}, \tilde{\beta} \in k^{\times}$.
Let us find a nowhere vanishing function $\phi_h$ that satisfies
\begin{align*}
	\frac{\phi_h \phi_{h+w}}{\phi_{h+v_1} \phi_{h+u_1}} &= \frac{\tilde{\alpha}}{\alpha_h}, \\
	\frac{\phi_h \phi_{h+w}}{\phi_{h+v_2} \phi_{h+u_2}} &= \frac{\tilde{\beta}}{\beta_h}.
\end{align*}
Since $v_1$ and $v_2$ are linearly independent and $\alpha_{h+v_2}\beta_h / \alpha_{h+w} \beta_{h+u_1} = \beta_{h+v_1}\alpha_h / \beta_{h+w}\alpha_{h+v_2}$, the two equations
\begin{align*}
	P_{h+v_1} &= \frac{\alpha_h}{\alpha_{h+u_2}}P_h, \\
	P_{h+v_2} &= \frac{\beta_h}{\beta_{h+u_1}}P_h
\end{align*}
are compatible and thus we can find $P_h \in k^{\times}$ satisfying the above relations.
Because of this particular construction of $P_h$, the two equations
\begin{align*}
	A_{h+v_1} &= \frac{\tilde{\alpha}}{\alpha_h}A_h, \\
	A_{h+u_2} &= P_h A_h
\end{align*}
are compatible and there exists a nowhere vanishing function $A_h$ that satisfies these two relations.
In the same way, we can find a nowhere vanishing function $B_h$ that satisfies
\begin{align*}
	B_{h+v_2} &= \frac{\tilde{\beta}}{\beta_h}B_h, \\
	B_{h+u_1} &= P_h B_h.
\end{align*}

In terms of $A_h$ and $B_h$ constructed above, the two equations
\begin{align*}
	\phi_{h+u_1} &= A_h \phi_h, \\
	\phi_{h+u_2} &= B_h \phi_h.
\end{align*}
are compatible and there exists a non-vanishing function $\phi_h$, which is the function we wanted to construct, as:
\begin{align*}
	\frac{\phi_h \phi_{h+w}}{\phi_{h+v_1} \phi_{h+u_1}} &= \frac{A_{h+v_1}}{A_h} = \frac{\tilde{\alpha}}{\alpha_h}, \\
	\frac{\phi_h \phi_{h+w}}{\phi_{h+v_2} \phi_{h+u_2}} &= \frac{B_{h+v_2}}{B_h} = \frac{\tilde{\beta}}{\beta_h}.
\end{align*}

(3) $\Rightarrow$ (1) follows immediately from Theorem~\ref{thm:fz} and Proposition~\ref{prop:gauge}.
\end{proof}

\begin{proof}[Proof of Theorem~\ref{thm:mainbkp}]

(1) $\Rightarrow$ (2).
Let $h = (l, m, n) \in \mathbb{Z}^3$ and
\begin{align*}
	H	&=	\{ h - r_1 v_1 - r_2 v_2 - r_3 v_3 \, | \, r_1, r_2, r_3 \in \mathbb{Z}_{\ge -2} \} \\
	&=	\{ (l', m', n') \, | \, l' \ge l - 2, m' \ge m - 2, n' \ge n - 2 \}.
\end{align*}
Since $H \subset \mathbb{Z}^3$ is a good domain, the corresponding initial value problem
\[
	f_{h'} = \begin{cases}
		\dfrac{\alpha_{h'}f_{h'+v_1}f_{h'+u_1} + \beta_{h'}f_{h'+v_2}f_{h'+u_2} + \gamma_{h'}f_{h'+v_3}f_{h'+u_3}}{f_{h'+w}} & (h' \in H \setminus H_0) \\
		X_{h'} & (h' \in H_0)
	\end{cases}
\]
has the Laurent property.

Let
\begin{align*}
	F &= \alpha_{h}f_{h+v_1}f_{h+u_1} + \beta_{h}f_{h+v_2}f_{h+u_2} + \gamma_{h}f_{h+v_3}f_{h+u_3}, \\
	A &= k[X_{h_0}, X^{-1}_{h_0} \, | \, h_0 \in H_0], \\
	A' &= A / (f_{h+w}, f_{h+u_3}),
\end{align*}
where $(f_{h+w}, f_{h+u_3})$ is the ideal of $A$ generated by $f_{h+w}$ and $f_{h+u_3}$.
In a similar way as in the proof of Theorem~\ref{thm:mainhm}, we calculate $F$ in the ring $A'$.
Since $f_{h+w} = f_{h+u_3} = 0$, we have
\begin{align*}
	f_{h+v_1} &= \frac{\gamma_{h+v_1}f_{h+u_2}X_{h+v_1+u_3}}{X_{h+w+v_1}}, \\
	f_{h+v_2} &= \frac{\gamma_{h+v_2}f_{h+u_1}X_{h+v_2+u_3}}{X_{h+w+v_2}},
\end{align*}
and
\[
	F = \frac{f_{h+u_1}f_{h+u_2}}{X_{h+w+v_1}X_{h+w+v_2}}
	(\alpha_{h}\gamma_{h+v_1}X_{h+v_1+u_3}X_{h+w+v_2} + \beta_{h}\gamma_{h+v_2}X_{h+v_2+u_3}X_{h+w+v_1}).
\]
From the discrete BKP equation, we have
\[
	X_{h+v_1+u_3}X_{h+w+v_2} = - \frac{\beta_{h+u_3}}{\alpha_{h+u_3}}X_{h+v_2+u_3}X_{h+w+v_1},
\]
and thus
\[
	F = \frac{f_{h+u_1}f_{h+u_2}X_{h+v_2+u_3}}{X_{h+w+v_2}}\left(\beta_{h}\gamma_{h+v_2} - \frac{\alpha_h\beta_{h+u_3}\gamma_{h+v_1}}{\alpha_{h+u_3}}\right).
\]
Since the initial value problem has the Laurent property, $F$ must be $0$ in $A'$.
Hence we have
\[
	\beta_{h}\gamma_{h+v_2} - \frac{\alpha_h\beta_{h+u_3}\gamma_{h+v_1}}{\alpha_{h+u_3}} = 0.
\]
The other two relations are obtained by cyclic permutation of the indices.

(2) $\Rightarrow$ (3).
Suppose that $\alpha_h, \beta_h$ and $\gamma_h$ satisfy the relation in (2) and let $\tilde{\alpha}, \tilde{\beta}, \tilde{\gamma} \in k^{\times}$.
Let us find a nowhere vanishing function $\phi_h$ satisfying
\begin{align*}
	\frac{\phi_h \phi_{h+w}}{\phi_{h+v_1} \phi_{h+u_1}} &= \frac{\tilde{\alpha}}{\alpha_h}, \\
	\frac{\phi_h \phi_{h+w}}{\phi_{h+v_2} \phi_{h+u_2}} &= \frac{\tilde{\beta}}{\beta_h}, \\
	\frac{\phi_h \phi_{h+w}}{\phi_{h+v_3} \phi_{h+u_3}} &= \frac{\tilde{\gamma}}{\gamma_h}.
\end{align*}
As in the proof of Theorem~\ref{thm:mainhm}, it suffices to check certain compatibility conditions.

We can find nowhere vanishing functions $P_h, Q_h, R_h$ on $\mathbb{Z}^3$ that satisfy
\begin{align*}
	P_{h+v_2} &= \frac{\beta_h}{\beta_{h+u_3}}P_h, &
	Q_{h+v_1} &= \frac{\alpha_h}{\alpha_{h+u_3}}Q_h, &
	R_{h+v_1} &= \frac{\alpha_h}{\alpha_{h+u_2}}R_h, \\
	P_{h+v_3} &= \frac{\gamma_h}{\gamma_{h+u_2}}P_h, &
	Q_{h+v_3} &= \frac{\gamma_h}{\gamma_{h+u_1}}Q_h, &
	R_{h+v_2} &= \frac{\beta_h}{\beta_{h+u_1}}R_h,
\end{align*}
since the compatibility conditions follow from the relations among $\alpha_h, \beta_h, \gamma_h$.
For example, the compatibility condition concerning $P_h$ can be checked as follows:
\begin{align*}
	\frac{\beta_{h+v_3}\beta_{h+u_3}\gamma_{h}\gamma_{h+w}}{\beta_{h}\beta_{h+w}\gamma_{h+v_2}\gamma_{h+u_2}} &= \frac{\gamma_{h+w}}{\beta_{h+w}}\frac{\beta_{h+v_3}\gamma_{h}}{\gamma_{h+u_2}}\frac{\beta_{h+u_3}}{\beta_{h}\gamma_{h+v_2}} \\
	&= \frac{\gamma_{h+w}}{\beta_{h+w}}\frac{\alpha_{h}\beta_{h+v_1}}{\alpha_{h+u_2}}\frac{\alpha_{h+u_3}}{\alpha_{h}\gamma_{h+v_1}} \\
	&= \frac{\alpha_{h+v_1+v_2}\beta_{h+v_1}\gamma_{h+v_1+u_1}}{\alpha_{h+v_1+v_3}\beta_{h+v_1+u_1}\gamma_{h+v_1}} \\
	&= 1.
\end{align*}

There exist non-vanishing functions $A_h, B_h, C_h$ that satisfy
\begin{align*}
	A_{h+v_1} &= \frac{\tilde{\alpha}}{\alpha_h}A_h, &
	B_{h+v_2} &= \frac{\tilde{\beta}}{\beta_h}B_h, &
	C_{h+v_3} &= \frac{\tilde{\gamma}}{\gamma_h}C_h, \\
	A_{h+u_2} &= R_h A_h, &
	B_{h+u_1} &= R_h B_h, &
	C_{h+u_1} &= Q_h C_h, \\
	A_{h+u_3} &= Q_h A_h, &
	B_{h+u_3} &= P_h B_h, &
	C_{h+u_2} &= P_h C_h,
\end{align*}
since these equations are compatible because of the construction of $P_h, Q_h$ and $R_h$.

Finally we can find a nowhere vanishing function $\phi_h$ that satisfies
\begin{align*}
	\phi_{h+u_1} &= A_h \phi_h, \\
	\phi_{h+u_2} &= B_h \phi_h, \\
	\phi_{h+u_3} &= C_h \phi_h,
\end{align*}
since the compatibility of these equations follows from the construction of $A_h, B_h$ and $C_h$.
This function $\phi_h$ is the one we set out to find.

(3) $\Rightarrow$ (1).
Immediate from Theorem~\ref{thm:fz} and Proposition~\ref{prop:gauge}.
\end{proof}

\begin{corollary}\label{cor:irred}
If the nonautonomous Hirota-Miwa equation or the nonautonomous discrete BKP equation has the Laurent property, then every iterate (except the initial values themselves which are units) is irreducible as a Laurent polynomial of the initial values.
\end{corollary}
\begin{proof}
As explained in \textsection\ref{sec:gauge}, gauge transformations preserve the irreducibility of the iterates.
\end{proof}

As in the case of the above equations, it is possible to study the conditions for the Laurent property for other nonautonomous bilinear equations.

\begin{proposition}\label{prop:reduchm}
A reduction of the nonautonomous Hirota-Miwa equation
\[
	f_h = \frac{\alpha_h f_{h+v_1}f_{h+u_1} + \beta_h f_{h+v_2}f_{h+u_2}}{f_{h+w}}
\]
has the Laurent property if and only if $\alpha_h, \beta_h$ satisfy
\begin{equation}\label{eq:aabb}
	\alpha_{h}\alpha_{h+w}\beta_{h+v_1}\beta_{h+u_1} = \beta_{h}\beta_{h+w}\alpha_{h+v_2}\alpha_{h+u_2}
\end{equation}
for all $h \in H$.
\end{proposition}
\begin{proof}
The ``if'' part follows from Proposition~\ref{prop:reduc}.
The converse is the same as in the proof of Theorem~\ref{thm:mainhm}.
\end{proof}

The Hirota-Miwa equation has infinitely many reductions.
However, the condition for the Laurent property for any possible reduction always has the same form as long as we denote shifts by $v_i, u_i, w$.
This is an important advantage of our notation.

Naturally, a similar property holds in the case of reductions of the discrete BKP equation.

\begin{proposition}
A reduction of the nonautonomous discrete BKP equation
\[
	f_h = \frac{\alpha_h f_{h+v_1}f_{h+u_1} + \beta_h f_{h+v_2}f_{h+u_2} + \gamma_h f_{h+v_3}f_{h+u_3}}{f_{h+w}}
\]
has the Laurent property if and only if $\alpha_h, \beta_h, \gamma_h$ satisfy the following three relations:
\begin{align*}
	\alpha_{h+v_2}\beta_{h}\gamma_{h+u_1} &= \alpha_{h+v_3}\beta_{h+u_1}\gamma_{h}, \\
	\alpha_{h+u_2}\beta_{h+v_3}\gamma_{h} &= \alpha_{h}\beta_{h+v_1}\gamma_{h+u_2}, \\
	\alpha_{h}\beta_{h+u_3}\gamma_{h+v_1} &= \alpha_{h+u_3}\beta_{h}\gamma_{h+v_2}.
\end{align*}
\end{proposition}

It should be noted that the condition (\ref{eq:aabbhm}) for the nonautonomous Hirita-Miwa equation to possess the Laurent property, coincides with the condition found in \cite{nonhm1} for it to pass the singularity confinement test.
In this sense, condition (\ref{eq:aabbhm}) can be regarded as a condition for the integrability of the nonautonomous Hirota-Miwa equation.


\section{Structure of denominators and algebraic entropy}\label{sec:deg}

In this section, we study the denominators of the solutions to equations with the Laurent property.
One aim of this section is to calculate the algebraic entropy of the equation, which is an important integrability criterion for discrete systems defined on a one-dimensional lattice.
Except in Theorem~\ref{thm:ent}, we consider only the case where the base field is $\mathbb{R}$ since Lemma~\ref{lem:posi} only holds over $\mathbb{R}$.

\begin{definition}[algebraic entropy \cite{entropy}]
Consider a discrete equation defined by a rational function on a one-dimensional lattice.
Let $f_0, \ldots, f_l$ be its initial values and let $(f_m)_{m \ge 0}$ be the solution.
The algebraic entropy of this equation is
\[
	\lim_{m \to +\infty}\frac{1}{m}\log(\deg f_m),
\]
where $\deg f_m$ stands for the degree of $f_m$ as a rational function of $f_0, \ldots, f_l$.
\end{definition}

Any equation with zero algebraic entropy is said to be integrable.

\begin{definition}
Let $f$ be a rational function of $z = (z_1, z_2, \cdots)$ over $\mathbb{R}$.
We shall say that $f$ is positive (resp.\ nonnegative) if for every sequence of positive real numbers $a = (a_1, a_2, \cdots)$, $f(a)$ can be defined as a real number and is positive (resp.\ nonnegative).
\end{definition}

It is clear that the sum, product and quotient of positive functions are again positive.

\begin{lemma}\label{lem:posi}
Let $f, g$ be positive polynomials of $z = (z_1, z_2, \cdots)$.
If $f + g$ is divisible by $z_j$, then so are $f$ and $g$.
\end{lemma}
\begin{proof}
We may assume $j = 1$.
Let us decompose $f$ and $g$ into
\begin{align*}
	f &= z_1 f_1 + f_2, \\
	g &= z_1 g_1 + g_2,
\end{align*}
where $f_2$ and $g_2$ do not depend on $z_1$.
Since $f + g$ is divisible by $z_1$, we have
\[
	0 = (f + g)\big|_{z_1 = 0} = f_2 + g_2
\]
and $f_2 = - g_2$.
Thus, it is sufficient to show that $f_2$ and $g_2$ are nonnegative.

Let $a = (a_1, a_2, \cdots)$ be an arbitrary sequence of positive real numbers and let $\tilde{a} = (\varepsilon, a_2, a_3, \cdots)$ for a small positive real number $\varepsilon$.
Since $f_2$ does not depend on $z_1$, we have $f_2(a) = f_2(\tilde{a})$.
The positivity of $f$ implies
\[
	0 < f(\tilde{a}) = \varepsilon f_1(\tilde{a}) + f_2(\tilde{a}) = \varepsilon f_1(\tilde{a}) + f_2(a)
\]
and, by taking the limit $\varepsilon \to +0$, we have
\[
	0 \le f_2(a).
\]
The positivity of $g_2$ can be shown in the same way.
\end{proof}

\begin{example}\label{exa:ex4}
Consider the equation
\[
\left\{\begin{array}{l}
	f_{m} = \dfrac{\alpha_m f^2_{m-1} + \beta_m}{f_{m-2}}, \\
	f_0 = X, f_1 = Y,
\end{array}\right.
\]
where $\alpha_m, \beta_m$ are positive real numbers satisfying $\alpha_m \alpha_{m-2}\beta^2_{m-1} = \beta_m \beta_{m-2}$.
This equation has the Laurent property (Example~\ref{exa:ex3}).
Let us decompose $f_m$ into $f_m = p_m / q_m$, 
where $p_m$ is a polynomial in $X$ and $Y$, and where $q_m$ is a monomial in $X$ and $Y$ with coefficient $1$; $p_m$ and $q_m$ are coprime as polynomials.
It is clear that $f_m, p_m, q_m$ are all positive.
Let us show that $q_m = X^{m-1}Y^{m-2}$ for $m \ge 2$.

First we show by induction that $p_m$ cannot be divided by $X$ or $Y$, and that $q_m$ can be divided by $q_{m-1}$ for $m \ge 2$.
A direct calculation shows that the statements are true for $m = 2, 3$.
The Laurent property of the equation and the expression
\[
	f_m = \frac{\alpha_m p^2_{m-1} + \beta_m q^2_{m-1}}{p_{m-2}} \frac{q_{m-2}}{q^2_{m-1}}
\]
imply that $(\alpha_m p^2_{m-1} + \beta_m q^2_{m-1}) / p_{m-2}$ is a Laurent polynomial.
It is easy to see that this is in fact a polynomial.
Since $f_m, p_m, q_m$ are positive and $p_{m-1}$ has no monomial factor, it follows from Lemma~\ref{lem:posi} that $\alpha_m p^2_{m-1} + \beta_m q^2_{m-1}$ has no monomial factor, either.
Since $p_{m-2}$ has no monomial factor and $q^2_{m-1} / q_{m-2}$ is a monomial, we obtain
\begin{align*}
	p_m &= \frac{\alpha_m p^2_{m-1} + \beta_m q^2_{m-1}}{p_{m-2}}, \\
	q_m &= \frac{q^2_{m-1}}{q_{m-2}}
\end{align*}
for $m \ge 4$.
Therefore, the statements are true for $m \ge 4$.

It follows from the above equation for $q_m$ that $q_m = X^{m-1}Y^{m-2}$ for $m \ge 4$.
Since $\deg p_m = \deg q_m + 1$, we have $\deg f_m = 2m - 2$.
Hence, the algebraic entropy of this equation is $0$.
\end{example}

\begin{example}
Let $r$ be an integer greater than $2$ and consider the equation
\[
\left\{\begin{array}{l}
	f_{m} = \dfrac{f^r_{m-1} + 1}{f_{m-2}}, \\
	f_0 = X, f_1 = Y.
\end{array}\right.
\]
This system has the Laurent property (Example~\ref{exa:ex3}).
Let us denote $f_m = p_m / q_m$ in the same way as in Example~\ref{exa:ex4}.
Then, it can be easily shown that
\begin{align*}
	q_m &= \frac{q^r_{m-1}}{q_{m-2}}, &
	p_m &= \frac{p^r_{m-1} + q^r_{m-1}}{p_{m-2}}, \\
	q_2 &= X, &
	p_2 &= Y^r + 1, \\
	q_3 &= X^r Y, &
	p_3 &= (Y^r + 1)^r + X^r.
\end{align*}
Thus we have
\begin{align*}
	\deg p_m &= \mathcal{O}(\lambda^m), \\
	\deg q_m &= \mathcal{O}(\lambda^m),
\end{align*}
where $\lambda = (r + \sqrt{r^2 - 4}) / 2$.
Hence, the algebraic entropy is $\log \lambda > 0$ and the equation is non-integrable.
\end{example}

Next we consider the denominators of the solutions to discrete bilinear equations that have the Laurent property.
Let $\alpha^{(i)}_h \in \mathbb{R}_{>0}$ and let
\begin{equation}
	f_h = \begin{cases}
		\dfrac{\alpha^{(1)}_h f_{h+u_1}f_{h+u_2} + \cdots + \alpha^{(n)}_h f_{h+u_n}f_{h+u_n}}{f_{h+w}} & (h \in H \setminus H_0) \\
		X_h & (h \in H_0)
	\end{cases}
\end{equation}
be an initial value problem for a discrete bilinear equation with the Laurent property.
Let us decompose $f_h$ as $f_h = p_h / q_h$ as before.
It is clear that $f_h, p_h, q_h$ are all positive.

\begin{lemma}\label{lem:denom}
$q_h$ satisfies the following relation:
\begin{equation}\label{eq:denom}
	q_h = \begin{cases}
		1	&	(h \in H_0), \\
		X_{h+w}\operatorname{LCM}\limits_{1 \le j \le n}(q_{h+v_j}q_{h+u_j})	&	(h + w \in H_0), \\
		\operatorname{LCM}_{1 \le j \le n}(q_{h+v_j}q_{h+u_j}) / q_{h+w}	&	(\text{otherwise}),
	\end{cases}
\end{equation}
where $\operatorname{LCM}$ stands for the least common multiple as a monomial.
In particular, $d^{(h_0)}_{h} = \deg_{X_{h_0}}q_{h}$ satisfies the following $(\max, +)$-equation:
\begin{equation}\label{eq:dh}
	d^{(h_0)}_h = \begin{cases}
		0	&	(h \in H_0), \\
		1 	&	(h + w = h_0), \\
		\max\limits_{1 \le j \le n}(d^{(h_0)}_{h+v_j} + d^{(h_0)}_{h+u_j})	&	(h + w \in H_0, h + w \ne h_0), \\
		\max\limits_{1 \le j \le n}(d^{(h_0)}_{h+v_j} + d^{(h_0)}_{h+u_j}) - d^{(h_0)}_{h+w}	&	(\text{otherwise}).
	\end{cases}
\end{equation}
\end{lemma}
\begin{proof}
Let $\le$ be the semi-order on $H$ defined in Definition~\ref{defi:order}.
We show the following four statements, by induction on $h \in H$:
\begin{itemize}
\item[(1)]
If $h' \le h$, then $q_{h'}$ is divisible by $q_h$.
\item[(2)]
If $h' \le h$ and $h' \ne h$, then $q_h \ne q_{h'}$.
\item[(3)]
$p_h$ cannot be divided by $X_{h_0}$ unless $h = h_0 \in H_0$.
\item[(4)]
$q_h$ satisfies (\ref{eq:denom}).
\end{itemize}

Since the case $h \in H_0$ is trivial, we may assume $h \in H \setminus H_0$.

First consider the case $h + w \in H_0$.
Let
\begin{align*}
	G &= \operatorname{LCM}_{j}(q_{h+v_j}q_{h+u_j}), \\
	G_j &= \frac{G}{q_{h+v_j}q_{h+u_j}}, \\
	F &= \alpha^{(1)}_h p_{h+v_1}p_{h+u_1} G_1 + \cdots + \alpha^{(n)}_h p_{h+v_n}p_{h+u_n} G_n.
\end{align*}
Then $f_h = F / (X_{h+w}G)$.
We will show that $F$ cannot be divided by $X_{h_0}$, for any $h_0 \in H_0$.
Assume that $X_{h_0}$ divides $F$.
Lemma~\ref{lem:posi} implies that $X_{h_0}$ divides $\alpha^{(j)}_h p_{h+v_j}p_{h+u_j} G_j$ for all $j$ since $\alpha^{(j)}_h p_{h+v_j}p_{h+u_j} G_j$ are positive polynomials.
If $h_0 \notin \{ h + v_1, \ldots, h + u_n \}$, then the induction hypothesis (3) implies that $p_{h+v_j}$ and $p_{h+u_j}$ cannot be divided by $X_{h_0}$.
Therefore $X_{h_0}$ must divide all $G_j$, which contradicts the definition of $G$.
On the other hand, if $h_0 = h + v_1$, then $X_{h_0}$ does not divide $p_{h+v_2}p_{h+u_2}$.
Thus $G_2$ is divisible by $X_{h_0}$ and so is $G$ since $G = G_2 q_{h+v_2}q_{h+u_2}$.
However, the induction hypothesis (4) implies that $X_{h_0}$ divides none of $q_{h+v_1}, \ldots, q_{h+u_n}$, which leads to a contradiction.
Since $X_{h_0}$ therefore cannot divide $F$, we have
\begin{align*}
	p_h &= F, \\
	q_h &= X_{h+w}G,
\end{align*}
which shows that (3) and (4) are true.
(1) and (2) follow immediately from the above expressions and from the induction hypothesis (1).

Next, we consider the case $h + w \in H \setminus H_0$.
Let
\begin{align*}
	G &= \frac{\operatorname{LCM}_{j}(q_{h+v_j}q_{h+u_j})}{q_{h+w}}, \\
	G_j &= \frac{q_{h+w}G}{q_{h+v_j}q_{h+u_j}}, \\
	F &= \alpha^{(1)}_h p_{h+v_1}p_{h+u_1} G_1 + \cdots + \alpha^{(n)}_h p_{h+v_n}p_{h+u_n} G_n.
\end{align*}
These are all polynomials since the induction hypothesis implies that $q_{h+w}$ divides $q_{h+v_1}, \ldots, q_{h+u_n}$.
Thus we have $f_h = F / (p_{h+w}G)$.
As in the case $h + w \in H_0$, $F$ cannot be divided by $X_{h_0}$.
Furthermore, it follows from the Laurent property of the equation that $F / p_{h+w}$ is a Laurent polynomial.
Since the induction hypothesis (3) implies that $p_{h+w}$ has no monomial factor, we have
\begin{align*}
	p_h &= \frac{F}{p_{h+w}}, \\
	q_h &= G,
\end{align*}
which shows that (3) and (4) are true;
(1) and (2) immediately follow from the above expressions and from the induction hypotheses.

Considering the degree of each initial value, we obtain the equation (\ref{eq:dh}).
\end{proof}

\begin{lemma}\label{lem:hmbkp}
In the case of the Hirota-Miwa equation and the discrete BKP equation, $q_h$ can be written explicitly as
\[
	q_h = \prod_{h_0 \in H_0, h_0 \le h + w}X_{h_0},
\]
where $\le$ is the semi-order defined in Definition~\ref{defi:order}.
\end{lemma}
\begin{proof}
Let $h_0 \in H_0$ and $d_h = d^{(h_0)}_h = \deg_{X_{h_0}}q_h$.
Then, $d_h$ satisfies (\ref{eq:dh}).
Let us show that
\[
	d_h = \begin{cases}
		1 	&	(h \ge h_0 - w), \\
		0	&	(\text{otherwise}).
	\end{cases}	
\]

Since an easy induction shows that $d_h = 0$ unless $h \ge h_0 - w$, it is sufficient to show (by induction on $h \in H$) that $d_h = 1$ for $h \ge h_0 - w$.
By (\ref{eq:dh}), we have
\[
	d_h = \max\limits_{j}(d_{h+v_j} + d_{h+u_j}) - d_{h+w}.
\]

If $h$ satisfies $h + w \ge h_0 - w$, then $h + v_1, \ldots, h + u_n \ge h_0 - w$.
The induction hypothesis then implies $d_{h+v_1} = \cdots = d_{h+u_n} = d_{h+w} = 1$ and thus we have $d_h = 1$.

If $h$ does not satisfy $h + w \ge h_0 - w$, then $d_{h+w} = 0$.
Since $h \ge h_0 - w$ and $h \ne h_0 - w$, there exists $x \in \{ v_1, \ldots, u_n \}$ such that $h + x \ge h_0 - w$.
The induction hypothesis in that case implies $d_{h+x} = 1$ and thus we have $d_h \ge 1$.
On the other hand, a direct coordinate calculation shows that if $x \in \mathbb{Z}^3$ satisfies $h + v_j \ge x$ and $h + u_j \ge x$ for some $j$, then $x$ satisfies $h + w \ge x$.
Therefore, for any $j$, either $h + v_j \ge h_0 - w$ or $h + u_j \ge h_0 - w$ does not hold.
Hence, we have $d_{h+v_j} + d_{h+u_j} \le 1$ and $d_h = 1$.
\end{proof}

\begin{theorem}\label{thm:ent}
Consider a discrete bilinear equation on a one-dimensional lattice that can be obtained as a reduction from the nonautonomous Hirota-Miwa equation or the nonautonomous discrete BKP equation, over any base field.
If this equation has the Laurent property, then its degree growth is at most quadratic.
In particular, its algebraic entropy is zero.
\end{theorem}
\begin{proof}
We only show the case of a reduction of the Hirota-Miwa equation since the proof remains valid for reductions of the discrete BKP equation.

A reduction of the Hirota-Miwa equation has the form
\begin{equation}\label{eq:gr1}
	f_m = \dfrac{\alpha_m f_{m-a}f_{m-l+a} + \beta_m f_{m-b}f_{m-l+b}}{f_{m-l}},
\end{equation}
where $0 < a < b < l$ are positive integers.
We may assume that $f_0 = X_0, \ldots, f_{l-1} = X_{l-1}$ are the initial values.

Let $\varphi \colon \mathbb{Z}^3 \to \mathbb{Z}$ be a $\mathbb{Z}$-linear map that gives a reduction from the Hirota-Miwa equation and $H = \varphi^{-1}(\mathbb{Z}_{\ge 0})$.
Then the initial domain for $H$ is $H_0 = \varphi^{-1}(\{ 0, \ldots, l-1 \})$.
Let $\gamma_h = \alpha_{\varphi(h)}, \delta_h = \beta_{\varphi(h)}$ and let
\begin{equation}\label{eq:gr2}
	g_{h} = \begin{cases}
		\dfrac{\gamma_{h} g_{h+v_1} g_{h+u_1} + \delta_{h} g_{h+v_2} g_{h+u_2}}{g_{h+w}} & (h \in H \setminus H_0), \\
		X_h & (h \in H_0)
	\end{cases}
\end{equation}
be the corresponding initial value problem.
It is sufficient to show that $\deg g_h$ has at most order $\mathcal{O}(m^2)$ for $\varphi(h) = m$ since Proposition~\ref{prop:reduc} implies that $\deg f_m \le \deg g_h$.

First let us reduce to the case $\gamma_h = \delta_h = \alpha_m = \beta_m = 1$.
Since (\ref{eq:gr1}) has the Laurent property, it follows from Proposition~\ref{prop:reduchm} that $\alpha_m$ and $\beta_m$ satisfy (\ref{eq:aabb}), where $v_1 = -a, u_1 = -l+a, v_2 = -b, u_2 = -l+b, w = -l$.
By construction, $\gamma_h$ and $\delta_h$ satisfy $(\ref{eq:aabbhm})$ and thus, by Theorem~\ref{thm:mainhm}, there exists a gauge transformation of the Hirota-Miwa equation which transforms $\alpha_h$ and $\beta_h$ to $1$.
Furthermore, as explained in \textsection\ref{sec:gauge}, a gauge transformation does not change the degree of the iterates, and without loss of generality, we can reduce to the case $\gamma_h = \delta_h = \alpha_m = \beta_m = 1$:
\begin{align*}
	f_m &= \dfrac{f_{m-a}f_{m-l+a} + f_{m-b}f_{m-l+b}}{f_{m-l}}, \\
	g_h &= \frac{g_{h+v_1} g_{h+u_1} + g_{h+v_2} g_{h+u_2}}{g_{h+w}}.
\end{align*}

Since the operation that reduces a Laurent polynomial by a prime number does not increase its degree, it is sufficient to prove the theorem under the condition that $k$ has characteristic $0$.
Therefore, we can use the lemmas shown in this section.

For $m_0 \in \{ 0, \ldots, l-1 \}$ and a sufficiently large integer $m$, we define $e^{(m_0)}_m$ by
\[
	e^{(m_0)}_m = \# \{ h_0 \in H_0 \, | \, h_0 \le h + w, \varphi(h_0) = m_0 \},
\]
where $\le$ is the semi-order on $\mathbb{Z}^3$ defined in Definition~\ref{defi:order} and $h \in \mathbb{Z}^3$ satisfies $\varphi(h) = m$.
Lemma~\ref{lem:hmbkp} implies that the degree of the denominator of $g_h$ coincides with
\[
	\sum^{l-1}_{m_0 = 1} e^{(m_0)}_m.
\]
Since the degree of the numerator of $g_h$ is always one higher than the degree of the denominator, it is sufficient to show that the growth of $e^{(m_0)}_m$ is at most quadratic.

Fix $m_0 \in \{ 0, \ldots, l-1 \}$ and denote $e^{(m_0)}_m$ by $e_m$.
Let $\varphi_{\mathbb{R}} \colon \mathbb{R}^3 \to \mathbb{R}$ be the $\mathbb{R}$-linear extension of $\varphi$, $P = \varphi^{-1}_{\mathbb{R}}(m_0) \subset \mathbb{R}^3$ and $\sigma = \operatorname{span}_{\mathbb{R}_{\ge 0}}(v_1, u_1, v_2, u_2) \subset \mathbb{R}^3$.
Then, the convex cone $h + w + \sigma$ is split by $P$ into two parts.
Let $A_h$ be the bounded part.
$A_h$ is a square pyramid with base $A_h \cap P$, and its height is proportional to $m - l - m_0$ (cf. Figure~\ref{fig:fig}).
Since
\[
	\{ h_0 \in H_0 \, | \, h_0 \le h + w, \varphi(h_0) = m_0 \} \subset A_h \cap P \cap \mathbb{Z}^3,
\]
$e_m$ is equal to or smaller than the number of lattice points contained in $A_h \cap P$.
Under the limit $m \to +\infty$, the number of lattice points in $A_h \cap P$ has the same growth as the area of $A_h \cap P$.
Since the height of $A_h$ has order $\mathcal{O}(m)$, the growth of the area of $A_h \cap P$ is quadratic.
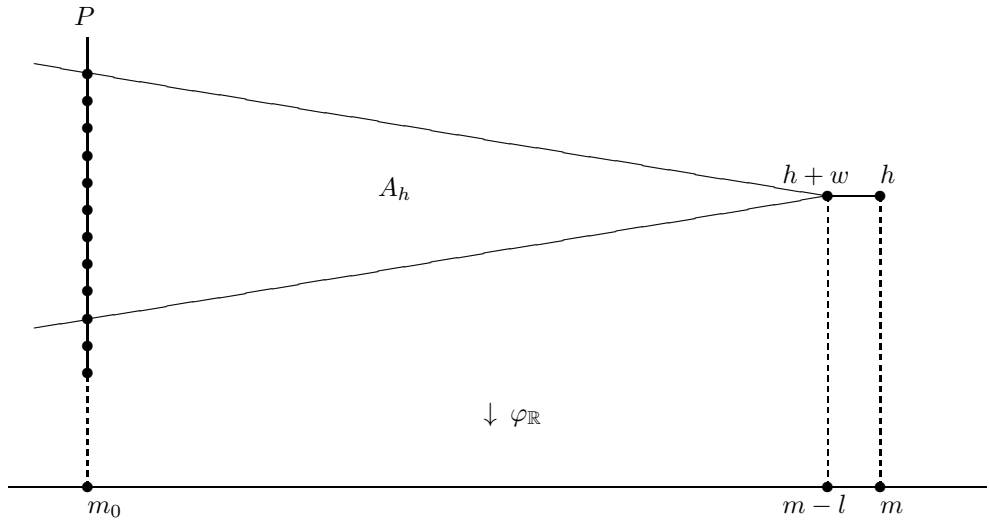
\begin{figure}
{\centering
\begin{picture}(400, 200)
	\put(20, 10){\line(1, 0){370}}
	
	\put(50, 10){\circle*{4}}
	\put(50, 50){\line(0, 1){130}}
	\multiput(50, 53)(0, 10.3){12}{\circle*{4}}
	\put(50, 0){$m_0$}
	
	\put(330, 10){\circle*{4}}
	\put(330, 120){\circle*{4}}
	\put(313, 125){$h+w$}
	\put(313, 0){$m-l$}
	
	\put(350, 10){\circle*{4}}
	\put(350, 120){\circle*{4}}
	\put(350, 0){$m$}
	\put(350, 125){$h$}
	
	\put(330, 120){\line(1, 0){20}}
	\put(330, 120){\line(-6, 1){300}}
	\put(330, 120){\line(-6, -1){300}}
	
	\multiput(50, 10)(0, 4){10}{\line(0, 1){2}}
	\multiput(330, 10)(0, 4){28}{\line(0, 1){2}}
	\multiput(350, 10)(0, 4){28}{\line(0, 1){2}}
	
	\put(160, 120){$A_h$}
	\put(45, 185){$P$}
	
	\put(200, 35){$\downarrow$}
	\put(210, 35){$\varphi_{\mathbb{R}}$}

\end{picture}
\par}
\caption{Schematic representation of the situation in which the degree growth of the equation is reduced to the number of lattice points contained in $A_h \cap P$.}\label{fig:fig}
\end{figure}
\end{proof}


\section{Conclusion}

In this paper we gave proofs of the theorems in \cite{rims} and discussed the Laurent property for nonautonomous systems.
First we gave elementary proofs of the Laurent property for the Hirota-Miwa equation and the discrete BKP equation.
We used the irreducibility and coprimeness of the iterates and we commented on them in \textsection\ref{sec:gauge} and Corollary~\ref{cor:irred}.
These concepts are thought to be closely related to integrability \cite{copr}.
Next we showed that a reduction and a gauge transformation of a discrete bilinear equation preserve the Laurent property.
Using these techniques we gave the explicit condition on the coefficients of discrete bilinear equations for them to possess the Laurent property.
Finally, we investigated the denominators of the iterates of an equation with the Laurent property and we calculated its algebraic entropy.
We showed that the degree of each initial value satisfies a $(\max, +)$-equation like (\ref{eq:dh}), which is in fact the ultra-discretization of the equation \cite{ud}.
Investigating this $(\max, +)$-equation and relying our results on reductions and gauge transformations, we showed that any reduction to a one-dimensional lattice of a nonautonomous Hirota-Miwa or discrete BKP equation that possesses the Laurent property, has zero algebraic entropy.

Throughout the paper we refrained from using the caterpillar lemma.
Although it is very powerful the caterpillar lemma is so complicated that we feel that when using it to show the Laurent property, we can hardly see the essential points of the proof.
Moreover, recent studies have shown that there exist many equations that have the Laurent property but are not amenable to the caterpillar lemma \cite{hk, exthv}.
Therefore, it seems to be important to investigate the Laurent property without recourse to this lemma.
Recently the usefulness of the Laurent property has been recognized in the field of discrete integrable systems.
One of the most interesting recent results is that many important systems, including discrete Painlev\'{e} equations, can be obtained as the coefficients of nonautonomous equations with the Laurent property \cite{hi, okubo}.
The Laurent property for such nonautonomous discrete systems, however, has not been well understood yet.
In the future, we intend to study general nonautonomous equations (not necessarily bilinear, even non-amenable to the caterpillar lemma) with the Laurent property.


\section*{acknowledgement}

I wish to thank Prof.~R. Willox for useful comments.
This work was partially supported by a Grant-in-Aid for Scientific Research of Japan Society for the Promotion of Science ($25 \cdot 3088$); and by the Program for Leading Graduate Schools, MEXT, Japan.


\end{document}